\documentclass[a4paper,USenglish]{lipics-v2021}
\usepackage[utf8]{inputenc}
\usepackage{cite}

\long\def\commentbegin #1\commentend{}

\usepackage{balance}
\usepackage{algorithm}
\usepackage[noend]{algpseudocode}
\usepackage{ifthen}
\usepackage{comment}
\usepackage{amsthm,amsmath}
\usepackage{graphicx}

\usepackage{pstricks,pst-node,pst-tree}

\usepackage{color}
\usepackage{mathrsfs}

\renewcommand{\epsilon}{\varepsilon}

\hideLIPIcs

\newcommand{\Prob}[1]{\hbox{\rm I\kern-2pt P}\left[#1\right]}

\DeclareMathAlphabet{\mathsc}{OT1}{cmr}{m}{sc}

\renewcommand{\geq}{\geqslant}
\renewcommand{\ge}{\geqslant}
\renewcommand{\leq}{\leqslant}

\newcommand{\REFLOW}{{\mathcal{R}_{\ell}}}
\newcommand{\REFHIGH}{{\mathcal{R}_{h}}}

\newboolean{short}
\setboolean{short}{false}

\algblockdefx[ExecBl]{BlockOn}{BlockOff}  [1]{#1}
  
\makeatletter
\ifthenelse{\equal{\ALG@noend}{t}}%
  {\algtext*{BlockOff}}
  {}%
\makeatother

\floatname{algorithm}{Procedure}

\newcommand{\shortOnly}[1]{\ifthenelse{\boolean{short}}{#1}{}}
\newcommand{\onlyShort}[1]{\ifthenelse{\boolean{short}}{}{#1}}
\newcommand{\longOnly}[1]{\ifthenelse{\boolean{short}}{}{#1}}
\newcommand{\onlyLong}[1]{\ifthenelse{\boolean{short}}{}{#1}}

\def\ShowComment{True}
\ifdefined\ShowComment
\def\billy#1{{\color{green}\underline{\textsf{Billy:}}} {\color{blue} \emph{#1}}}
\def\gopal#1{{\color{red}\underline{\textsf{Gopal:}}} {\color{blue} \emph{#1}}}
\def\david#1{{\color{orange}\underline{\bf{David:}}} {\color{blue} \hl{#1}}}
\def\shay#1{\hl{shay: #1}}
\else
\def\billy#1{}
\def\gopal#1{}
\def\david#1{}
\def\shay#1{}
\fi

\def\paragraphbf#1{\par\vskip 7pt\noindent\textbf{#1}\hskip 10pt}
\def\inline#1:{\par\vskip 7pt\noindent{\bf #1:}\hskip 10pt}

\def\candidateSTATE{\textsf{Candidate}}
\def\electedSTATE{\textsf{Elected}}
\def\LEADER{\textsc{leader}}
\def\nonelectedSTATE{\textsf{Non-elected}}

\def\RefNonselectedSTATE{\textsf{Non-selected}}
\def\RefReadySTATE{\textsf{Ready}}
\def\RefChosenSelectedSTATE{\textsf{Chosen-Selected}}
\def\RefInDisputeSTATE{\textsf{In-Dispute}}

\def\Chosen{\textsf{CHOSEN}}
\def\Contender{\textsf{CONTENDER}}
\def\CANDSTATE{\textsf{CAND-STATE}}
\def\REFSTATE{\textsf{REF-STATE}}
\def\MLIST{\textsf{M-List}}
\def\NUMREPLIES{\textsf{NUM-REPLIES}}
\def\SENDLIST{\textsf{Send-List}}

\def\polylog{\operatorname{polylog}}

\def\REQUEST{\textsc{request}}
\def\APPROVED{\textsc{approved}}

\def\DECLINED{\textsc{declined}}
\def\LOSES{\textsc{loses}}
\def\DECIDE{\textsc{dispute}}

\def\WAKEUP{\textsc{wakeup}}

\def\OnReceiveMessage{\textbf{On\_Receive\_Message}}
\def\Initialize{\textbf{Initialize}}
\def\Candidate{\textbf{Candidate}}
\def\CandidateDecideResponse{\textbf{Candidate\_Dispute\_Response}}
\def\Referee{\textbf{Referee}}
\def\RefereeRequestResponse{\textbf{Referee\_Request\_Response}}
\def\RefereeDisputeResponse{\textbf{Referee\_Dispute\_Reply\_Response}}

\def\SendMessage{\textbf{Send\_Message}}

\def\ID{RANK}

\renewcommand{\geq}{\geqslant}
\renewcommand{\ge}{\geqslant}
\renewcommand{\leq}{\leqslant}

\long\def\hide #1\hideend{}

\newcommand{\squishlist}{
 \begin{list}{$\bullet$}
  { \setlength{\itemsep}{0pt}
     \setlength{\parsep}{3pt}
     \setlength{\topsep}{3pt}
     \setlength{\partopsep}{0pt}
     \setlength{\leftmargin}{1.5em}
     \setlength{\labelwidth}{1em}
     \setlength{\labelsep}{0.5em} } }

\newcommand{\squishlisttwo}{
 \begin{list}{$\bullet$}
  { \setlength{\itemsep}{0pt}
     \setlength{\parsep}{0pt}
    \setlength{\topsep}{0pt}
    \setlength{\partopsep}{0pt}
    \setlength{\leftmargin}{2em}
    \setlength{\labelwidth}{1.5em}
    \setlength{\labelsep}{0.5em} } }

\newcommand{\squishend}{
  \end{list}  }

\title{Singularly Near Optimal Leader Election in Asynchronous Networks}

\titlerunning{Singularly Near Optimal Leader Election in Asynchronous Networks}

\author{Shay Kutten}{Faculty of Industrial Engineering and Management, Technion - Israel Institute of Technology, Haifa, Israel}{kutten@technion.ac.il}{0000-0003-2062-6855}{This work was supported in part by the Bi-national Science Foundation (BSF) grant 2016419.}

\author{William K. {Moses Jr.}}{Department of Computer Science, University of Houston, Houston, TX, USA}{wkmjr3@gmail.com}{0000-0002-4533-7593}{Part of this work was done while the author was a postdoc at the Technion - Israel Institute of Technology in Israel. This work was supported in part by a Technion fellowship and in part by NSF grants, CCF1540512, IIS-1633720, CCF-1717075, and BSF grant 2016419.}

\author{Gopal Pandurangan}{Department of Computer Science, University of Houston, Houston, TX, USA}{gopal@cs.uh.edu}{0000-0001-5833-6592}{This work was supported in part by NSF grants CCF-1717075, CCF-1540512, IIS-1633720, and BSF grant 2016419.}

\author{David Peleg}{Department of Computer Science and Applied Mathematics, Weizmann Institute of Science, Rehovot, Israel}{david.peleg@weizmann.ac.il}{0000-0003-1590-0506}{This work was supported in part by the US-Israel Binational Science Foundation grant 2016732.}

\authorrunning{S. Kutten, W.\,K. Moses Jr., G. Pandurangan, and D. Peleg}
\Copyright{Shay Kutten, William K. Moses Jr., Gopal Pandurangan, and David Peleg}



\EventEditors{Seth Gilbert}
\EventNoEds{1}
\EventLongTitle{35th International Symposium on Distributed Computing (DISC 2021)}
\EventShortTitle{DISC 2021}
\EventAcronym{DISC}
\EventYear{2021}
\EventDate{October 4--8, 2021}
\EventLocation{Freiburg, Germany (Virtual Conference)}
\EventLogo{}
\SeriesVolume{209}
\ArticleNo{13}

\nolinenumbers

\begin{document}
\maketitle
\begin{abstract}
This paper concerns designing  distributed algorithms that are \emph{singularly optimal}, i.e., algorithms that are \emph{simultaneously}  time and message \emph{optimal}, for  the  fundamental  leader election problem in \emph{asynchronous}  networks.  

Kutten et al. (JACM 2015) presented a singularly near optimal randomized leader election algorithm for general \emph{synchronous} networks that ran in $O(D)$ time and used $O(m \log n)$
messages (where $D$, $m$, and $n$ are the network's diameter, number of edges and number of nodes, respectively) with high probability.\footnote{Throughout, ``with high probability'' means  ``with probability at least $1-1/n^c$, for constant $c$.''} Both  bounds are near optimal (up to a logarithmic factor), since $\Omega(D)$ and $\Omega(m)$ are the respective lower bounds for time and messages for leader election even for synchronous networks and even for (Monte-Carlo) randomized algorithms.
On the other hand, for general asynchronous networks, leader election algorithms are only known that are either time or message optimal, but not both. Kutten et al. (DISC 2020) presented a randomized  asynchronous leader election algorithm that is singularly near optimal for \emph{complete networks}, but left open the problem for general networks.

This paper shows that singularly near optimal (up to polylogarithmic factors) bounds can be achieved for general \emph{asynchronous} networks. 
We present a randomized singularly near optimal leader election algorithm that runs in $O(D  + \log^2n)$ time and $O(m\log^2 n)$ messages with high probability.  
Our result is the first known distributed leader election algorithm for asynchronous networks that is near optimal with respect to both time and message complexity and  improves over a long line of results including the classical results of Gallager et al. (ACM TOPLAS, 1983), Peleg (JPDC, 1989), and Awerbuch (STOC 89).

\end{abstract}


\section{Introduction}
\label{sec:intro}

\paragraphbf{Background and motivation.}
Trade-offs between resource bounds (typically time and space) form a major subject of study in classical theory of computation. In distributed computing, it is common to focus on two fundamental measures, the time and message complexity of a distributed network algorithm, and trade-offs between time and 
communication have been well studied. See, e.g., \cite{A89,awerbuch1985distributed,awerbuch1985complexity,AG91} for early trade-offs. 
A question that arose
more recently
regarding various distributed problems is 
whether the problem admits an algorithm that is optimal 
in time and communication simultaneously. In \cite{stoc17,gmyr}, such algorithms are called \emph{singularly optimal} (or \emph{singularly ``near optimal''} for algorithms whose complexity is polylogarithmically worse than optimal). 
Such algorithms have been shown in recent years for
minimum spanning tree,
 (approximate) shortest paths, leader election, and several other 
 problems \cite{KPPRT15jacm,stoc17,elkin2017simple,haeupler2018round}.
 
All the above results were shown for \emph{synchronous} networks, and \emph{do not} apply to asynchronous networks. 
In particular, singularly near optimal randomized \emph{synchronous} leader election algorithms were presented in \cite{KPPRT15jacm}.
These algorithms require $O(D)$ time and use $O(m \log n)$ messages with high probability (where $D$, $n$ and $m$ are the network's diameter, number of nodes and number of edges, respectively). Singular near optimality follows from the fact
that $\Omega(D)$ and $\Omega(m)$ are lower bounds for time and messages for leader election even for synchronous networks and even for randomized algorithms \cite{KPPRT15jacm}.
These algorithms inherently rely on the synchronous communication, so an attempt to convert them to asynchronous networks would probably 
incur heavy cost overheads
(see the discussion of synchronizers below).
In fact, the question whether similar bounds can be achieved for general \emph{asynchronous} networks was left open, although one can obtain algorithms that are separately time optimal \cite{peleg-jpdc} or message optimal \cite{GallagerHS1983}.

A singularly near optimal randomized \emph{asynchronous} leader election algorithm was recently presented in \cite{disc2020}, but only for \emph{complete} networks. It requires $O(n)$ messages and $O(\log^2 n)$ time, which is singularly optimal (up to logarithmic factors) since $\Omega(n)$ and $\Omega(1)$ are the respective message and time lower bounds for leader election in complete $n$-node networks. 
The question whether 
leader election in general 
networks
admits an \emph{asynchronous} 
singularly 
near optimal 
algorithm or 
exhibits an inherent \emph{time-messages trade-off} was again left as an
open problem. The algorithm in \cite{disc2020} heavily utilizes the special nature of the complete graph, so designing an asynchronous algorithm for general graphs requires additional tools 
and insights. 

Leader election is 
a central and intensively studied problem
in distributed computing.  It captures the pivotal notion of symmetry breaking in 
contexts involving the entire
network (``global algorithms''). Given a leader, many other problems become trivial. 
Consequently,
a singularly optimal leader election algorithm 
may
ease the design of a singularly optimal algorithms for many other tasks. In addition to its theoretical importance, leader election is used in multiple practical contexts. The literature is too rich to cover here, but see, for example \cite{Lann77-DistSystems, GallagerHS1983,AG91,KPPRT15jacm, Lyn96,peleg-jpdc,santoro-book,GerardTelDistributedAlgosBook,chandra2007paxos,ghemawat2003google,chang2008bigtable}.  

In our 
setting,
an \emph{arbitrary} subset of nodes can \emph{wake up spontaneously at arbitrary times} and start the election algorithm by sending messages over the network. The algorithm should terminate with a \emph{unique} node $v$ being elected as leader 
(where initially, all the nodes are in the same state, ``not leader'')
and the leader's identity should be \emph{known to all nodes}.
Our goal in this paper is to design  leader election algorithms in distributed \emph{asynchronous} networks that  are singularly near optimal.

\paragraphbf{Using synchronizers.}
One can convert a synchronous algorithm to work on an asynchronous network  using a standard tool known as  a \emph{synchronizer} \cite{awerbuch1985complexity};
however, such a conversion typically increases substantially either the time or the message complexity or both. Moreover, there is usually a 
non-negligible
cost associated with \emph{constructing}
such a synchronizer in the first place. For example, applying the simple $\alpha$ synchronizer (which does not require the a priori existence of a leader or a spanning tree)
to the singularly optimal synchronous leader election algorithm of \cite{KPPRT15jacm} yields an asynchronous algorithm with message complexity of $O(mD\log n)$ and  time complexity of $O(D)$; this algorithm is not message optimal, especially for large diameter networks.  Indeed, many prior works (see e.g., \cite{AP90}), do construct efficient synchronizers that can achieve optimal conversion from synchronous to asynchronous algorithms with respect to both time and messages, but  constructing the synchronizer itself requires a substantial preprocessing or initialization  cost. For example, the message cost of the synchronizer protocol of \cite{AP90} can be as high as $O(mn)$. Moreover, several synchronizer protocols, such as $\beta$ and $\gamma$ of \cite{awerbuch1985complexity} and that of \cite{AP90}, require the existence of a \emph{leader} or \emph{a spanning  tree};
hence these synchronizers are not useful for designing leader election algorithms. 
For these reasons,
designing singularly optimal algorithms is more
challenging for asynchronous networks than for synchronous ones and requires new approaches.

\paragraphbf{Distributed Computing Model.}
\label{sec:model}
We model a distributed network as an arbitrary undirected connected graph $G=(V,E)$, $|V|=n$, $|E| =m$, similar to the standard model of~\cite{GallagerHS1983,AG91,KorachKuttenMoran-ModularLE-TOPLAS,A89},
except that, in addition,  processors can access \emph{private unbiased coins}.
Nodes have only knowledge of themselves and do not have any knowledge of their neighbors and their identities (if any). This is the 
commonly used
\emph{clean network} or $KT_0$ model
(see e.g., \cite{vainish}).

Our algorithm does not require that nodes have unique identities; however, it requires
that nodes have knowledge of $n$, the network size, or at least a constant factor
approximation of $n$.  We note that several prior algorithms for leader election require knowledge of
$n$ \cite{A89,SS94,AM94}. 
If nodes have unique identifiers, we assume that they are of size $O(\log n)$ bits.

We assume the standard \emph{asynchronous} $\mathcal{CONGEST}$ communication model~\cite{peleg-locality}, where messages (each message is of $O(\log n)$ bits) sent over an edge incur unpredictable but finite delays, in an error-free and FIFO manner (i.e., messages will arrive in sequence).
However, for the sake of time analysis, it is assumed that message
takes \emph{at most one time unit} to be delivered across an edge.
As is usual, we assume that local computation within a node is instantaneous and free; however, our algorithm will involve  only lightweight local computations.

We follow the standard timing and wake-up assumptions used in prior asynchronous protocols (see \cite{AG91,GallagerHS1983,singh97}).  Nodes are initially asleep, and a node enters the execution when it is woken up by 
the environment 
or upon receiving messages from awakened neighbors. 
As usual, uncertainties in the environment can be modeled  by means of an \emph{adversary} that controls some of the execution parameters. Specifically, we assume an \emph{adversarial wake up} model, where node wake-up times are scheduled by an adversary (who may decide to keep some nodes dormant). 
The time complexity is measured from the moment the first node wakes up.
A node can also be woken up by receiving messages from other nodes. 
In addition to the wake-up schedule, the adversary also decides 
the time delay of each message.
These decisions are done \emph{adaptively}, i.e., when the adversary makes a decision to wake up a node or delay a message, it has access to the results of all previous coin flips. 
The above adversarial wakeup model  should be contrasted with the weaker \emph{simultaneous wake up} model, where all nodes are assumed to be awake at the beginning of computation; simultaneous wake up is typically assumed in the design of synchronous protocols (see e.g., \cite{AG91,KPPRT15tcs,KPPRT15jacm}).
In the asynchronous setting, once a node enters execution, it performs all the computations required of it by the algorithm, and sends out messages to neighbors as specified by the algorithm.

\paragraphbf{Our  Contribution.}
The main 
question addressed in this paper is whether singularly optimal bounds for leader election can be achieved for general \emph{asynchronous} networks. We answer this question in the affirmative and present 
    a randomized singularly near optimal leader election algorithm that elects a leader with high probability, runs in $O(D + \log^2n)$ time with high probability, and has message complexity $O(m\log^2 n)$ with high probability, where high probability is probability at least $1-1/n^c$, for constant $c$. Our algorithm even works in anonymous networks.
 To the best of our knowledge, this  is the first known distributed leader election algorithm for asynchronous general networks that is near optimal with respect to both time and message complexity and improves over a long line of results (see Table \ref{table:upper-bound-results})  including the classical results of Gallager et al.~\cite{GallagerHS1983}, Peleg~\cite{peleg-jpdc}, and Awerbuch~\cite{A89}. 
We refer to Table \ref{table:upper-bound-results} 
for a comparison of message and time complexity bounds of  leader election algorithms in asynchronous networks. It should be noted that
none of the prior results achieve time and message bounds that are simultaneously close to optimal bounds (even within a $O(\polylog n)$ factor) of $\Theta(D)$ (time) and $\Theta(m)$ (messages) respectively. We note our bounds
are almost as good as those obtained for the synchronous model: the work of Kutten et al.~\cite{KPPRT15jacm} presented a
$O(m \log n)$ messages (with high probability) and a $O(D)$ algorithm.
It is open whether one can design a (tight) singularly optimal leader election algorithm
that uses $O(m)$ messages and $O(D)$ time even for the \emph{synchronous} setting.

The importance of having a singularly optimal (or near optimal) leader election is that it can serve as a basic building block in designing singularly (or near) optimal
asynchronous algorithms for other fundamental problems such as MST and shortest paths. Currently, we are not aware of such algorithms for these problems
in \emph{asynchronous} networks (unlike synchronous networks~\cite{elkin2017simple,haeupler2018round,stoc17}).

\begin{table*}[ht]
	\caption{\small Comparison of leader election algorithms for a general graph with $n$ nodes, $m$ edges, and $D$ diameter in asynchronous systems along with our contributions. Note that if the algorithm was deterministic, then it was required that nodes have unique IDs. Randomized solutions may not have such a requirement. 
	} 
	\centering 
	\begin{tabular}{|c|c|c|c|}
		\hline
		Paper & Message Complexity & Time Complexity   & Type of Solution \\
		\hline
		\hline
		Gallager et al.~\cite{GallagerHS1983} & $O(m + n \log n)$ & $O(n \log n)$ & Deterministic\\
		\hline
		Lavall{\'e}e and Lavault~\cite{LL90}* & $O(m+ n \log n)$ & $O(n \log \log (n/\epsilon))$ & Randomized \\
		\hline
		Chin and Ting~\cite{chin1990improving} & $O(m + n \log n)$ & $O(n \log^* n)$ & Deterministic\\
		\hline
		Gafni~\cite{gafni-election} & $O(m + n \log n)$ & $O(n \log^* n)$ & Non-deterministic\\
		\hline
		Awerbuch~\cite{awerbuch-optimal}, Faloutsos and Molle~\cite{faloutsos}** & $O(m + n \log n)$ & $O(n)$  & Deterministic\\
		\hline
		Schieber and Snir~\cite{SS94} & $O(m + n \log n)$ & $O(n)$ & Randomized\\
		\hline
		Afek and Matias~\cite{AM94} & $O(m)$ & $O(n)$ & Randomized\\
		\hline
		Peleg~\cite{peleg-jpdc} & $O(mD)$ & $O(D)$ & Deterministic\\
		\hline
		Awerbuch~\cite{A89}*** & $O(m^{1+\epsilon})$ & $O(D^{1+\epsilon})$  & Deterministic\\
		\hline
		This paper & $O(m \log^2 n)$ & $O(D + \log^2 n)$ & Randomized\\
		\hline
		\multicolumn{4}{|l|}{*They claim a \textit{virtual} running time of $O(D' \log \log (n/\epsilon))$, corresponding to an algorithmically}\\
		\multicolumn{4}{|l|}{\hspace{1em}  constructed subgraph of diameter $D'$ of the initial network. $D'$ may be as large as $n$. $0 < \epsilon < 1$.}\\
		\hline
		\multicolumn{4}{|l|}{**The algorithm of~\cite{faloutsos} is a corrected version of the one in~\cite{awerbuch-optimal}.}\\
		\hline
		\multicolumn{4}{|l|}{***Here, $\epsilon$ can be any value $>0$.}\\
		\hline
	\end{tabular}
	\label{table:upper-bound-results}
\end{table*}

Our algorithm makes use of several elementary techniques, combined in a suitable way.
In particular, it exploits the idea of using groups of referees as \emph{quorums} 
in order to ensure mutual exclusion, an idea utilized in several papers, e.g.~\cite{CPR19,GRS18,disc2020,KPPRT15tcs,KrishnaRamanathan:randomized}. Traditionally, verifying that an entire quorum has been secured is achieved by counting the number of supporting referees. It should be noted, though, that such a counting process is problematic under the asynchronous communication model. This difficulty requires us to introduce some \emph{slack} to the quorum sizes, and rely on it in order to ensure quorum intersection with high probability.
Another interesting feature of this algorithm concerns the way it manages communication. Message transmissions are performed using \emph{flood-based broadcasts}, even when the message is targeted at a \emph{single} recipient (unlike most previous algorithms, where such messages are sent by unicast along a specific path). This is done since in an asynchronous network, paths defined by previous broadcasts might not be shortest, so using them later might prevent us from attaining a near diameter time. The obvious down side is that using broadcasts for transmitting individual messages is expensive in communication. Hence, one delicate technical point is how to maintain a tight cap on the overall number of wasteful broadcasts, in order to save on messages and on congestion. The key idea is to ensure that the number of active ``speakers'' (as opposed to passive ``relays'' who merely forward messages) during the entire execution is kept small (logarithmic in the network size).

We first
compare the technical contribution of this work 
to the 
known singularly near optimal \emph{synchronous} algorithm of Kutten et al.~\cite{KPPRT15jacm} for \emph{general graphs}. The task there is significantly easier since nodes know when to terminate: once a node stops receiving an echo, exactly one node (the one with the highest random rank) will know it is the leader. Moreover, in the synchronous setting, this happens after $O(D)$ rounds. In the asynchronous setting, this is not possible (unless some heavy message overhead is added, e.g., by using a synchronizer). This leads to technical challenges that we overcome utilizing randomization, broadcasts, and quorums. Randomization is used not only to reduce the number of messages (similar to the synchronous case), but more importantly to implement the quorums.

We also 
compare the technical contribution of this work 
to the 
known singularly near optimal \emph{asynchronous} algorithm of Kutten et al.~\cite{disc2020} for \emph{complete graphs}, 
which 
uses similar techniques. One key change
concerns the way candidates communicate with referees.
In ~\cite{disc2020}, this is done by sending messages directly, which 
is doable in complete graphs. In contrast,  
in the current algorithm the candidates must rely on broadcasts to send messages to referees. 
This change raises additional challenges and necessitates a major modification to the algorithm of~\cite{disc2020}.
Specifically, in that algorithm, 
each candidate selects, and hence knows, its referees in each phase 
(this is doable because the graph is complete). 
In contrast, in our algorithm for an arbitrary graph, a candidate does not know the referees (because referees are chosen independently of the candidate). As a result, it is necessary to keep an accurate estimate of the number of referees. 
We rely on this estimate to ensure 
the existence of exactly 
one leader, with high probability 
(if the estimate is too low, multiple nodes might become leaders; if it is too high, no one will become a leader).

Moreover, it should be stressed that the technique used in~\cite{disc2020} for saving on messages cannot be used here. 
Therein, 
candidates compete with each other in \emph{phases}, and get eliminated gradually, until a single candidate remains. Here, there does not appear to be a way for the candidate to save by sending information to only specific referees in a message optimal manner. 
Specifically, applying the algorithm of~\cite{disc2020} on a general graph 
by replacing direct communication with broadcasts (and making no other changes) would result in a high communication cost of $\Omega(mn)$ messages and possibly a higher run time (due to congestion), compared to the performance of the current algorithm.

\paragraphbf{Related Work.}
Leader election has been very well-studied in distributed networks for many decades.
Le~Lann~\cite{Lann77-DistSystems} first studied the problem in a 
ring network and  the seminal paper of Gallager, Humblet, and 
Spira~\cite{GallagerHS1983} studied it in
general graphs. Since then, various algorithms and lower bounds are known in different models with synchronous/asynchronous communication and in networks of varying topologies, such as cycles, complete graphs, or arbitrary graphs. See, e.g.,
\cite{KhanKMPT08,KorachKuttenMoran-ModularLE-TOPLAS,KPPRT15jacm,KPPRT15tcs, Lyn96,peleg-jpdc,santoro-book,GerardTelDistributedAlgosBook} 
and the references therein. 

Prior to this paper, there were no known singularly near optimal algorithms for \emph{general asynchronous} networks, i.e., algorithms that take $\tilde{O}(m)$ messages and $\tilde{O}(D)$ time.\footnote{$\tilde{O}$ notation hides a $\polylog(n)$ factor.}\footnote{There was, however, work done on a ring where Itai and Rodeh~\cite{itai-rodeh} presented an algorithm with $O(n \log n)$ bit complexity on expectation. Further analysis shows that the message complexity is $O(n \log n)$ on expectation and the time complexity is $O(n \log n)$ on expectation.} 
Gallager, Humblet, and Spira~\cite{GallagerHS1983} presented a minimum weight spanning tree (MST) algorithm (also applicable for leader election) with message complexity $O(m+n \log n)$ and time complexity $O(n \log n)$; this is (essentially)
message optimal~\cite{KPPRT15jacm} but not time optimal. 
Hence, further research concentrated on improving the time complexity of MST algorithms. The time complexity
was first improved to $O(n \log \log (n/\epsilon))$, $0 < \epsilon < 1$, by Lavall{\'e}e and Lavault~\cite{LL90}, then to $O(n \log^* n)$ independently by Chin and Ting~\cite{chin1990improving} and Gafni~\cite{gafni-election},
and finally to $O(n)$ by Awerbuch~\cite{awerbuch-optimal} (see also~\cite{faloutsos}). 
Thus using MST algorithms for leader election in asynchronous networks does not yield the best possible time complexity of $O(D)$.
Peleg's leader election algorithm \cite{peleg-jpdc}, on the other hand, takes  $O(D)$ time and uses $O(mD)$ messages; this is time optimal, but not message optimal.

Korach et al.~\cite{KorachPODC1984}, Humblet \cite{humblet-clique}, Peterson \cite{Pet84} and Afek and Gafni~\cite{AG91} presented $O(n \log n)$ message algorithms for \emph{asynchronous complete} networks. Korach, Kutten, and  Moran~\cite{KorachKuttenMoran-ModularLE-TOPLAS} presented a general method plus applications to various classes of graphs.  

For \emph{anonymous} networks under some reasonable assumptions, deterministic leader election was shown to be impossible, using symmetry arguments~\cite{AngluinSTOC80}. Randomization comes to the rescue in this case; random rank assignment is often used to assign unique identifiers, as done herein. Randomization also allows us to beat the lower bounds for deterministic algorithms, albeit at the risk of a small chance of error. As a starting step, Schieber and Snir~\cite{SS94} developed a randomized algorithm for leader election in anonymous asynchronous networks that took $O(m + n \log n)$ messages and $O(n)$ time. Afek and Matias~\cite{AM94} presented a randomized algorithm that took the same time but improved the message complexity to $O(m)$, which is message optimal.\footnote{Note that~\cite{AM94}'s message complexity is $O(m)$ when the solution succeeds with constant probability. For a solution that succeeds with high probability, the message complexity grows to $O(mn \log^2 n)$.}

For synchronous networks, Pandurangan et al.~\cite{stoc17} and Elkin~\cite{elkin2017simple} have presented singularly near optimal distributed algorithms for MST.
Note that
optimal time for MST means $\Omega(D+\sqrt{n})$ \cite{peleg1999near}, while for leader election in asynchronous networks, optimal time is $O(D)$, as shown in \cite{peleg-jpdc} but using worse message complexity.

We note that the singularly optimal algorithm of this paper (for asynchronous networks) as well as those of \cite{stoc17} and \cite{elkin2017simple} (for synchronous networks)
assume the so-called \emph{clean network model}, a.k.a.\ $KT_0$~\cite{peleg-locality} (see Section \ref{sec:model}), where nodes
do not have initial knowledge of the identity of their neighbors.
But the above optimal  results do not in general apply to the $KT_1$  model, where nodes have initial knowledge
of the identities of their neighbors.
Clearly, the distinction between $KT_0$ and $KT_1$ has  no bearing on the asymptotic bounds for the time complexity, 
but it is significant when considering message complexity.
Awerbuch et al.~\cite{vainish} show that $\Omega(m)$ is a message lower bound for broadcast (and hence for leader election and MST) in the $KT_1$ model,  
if one allows only (possibly randomized Monte Carlo) comparison-based algorithms, i.e., algorithms that can operate on IDs only by comparing them.
(We note that all  algorithms mentioned earlier are comparison-based, including ours.)  
Hence, the result of \cite{vainish} implies that our 
leader election algorithm (which is comparison-based and randomized) is \emph{time} and \emph{message near optimal}
in the $KT_1$ model if one considers comparison-based algorithms only.

On the other hand, for \emph{randomized non-comparison-based} algorithms,  the message lower bound
of $\Omega(m)$ does not apply in the $KT_1$ model. 
King et al.~\cite{KingKT15} presented a randomized, non-comparison-based Monte
Carlo algorithm in the $KT_1$ model for spanning tree, MST (and hence leader election) in $\tilde{O}(n)$ messages ($\Omega(n)$ is a message lower bound) and in $\tilde{O}(n)$ time (see also ~\cite{MashreghiK17}).
 While this algorithm shows that one can achieve $o(m)$ message complexity
(when $m = \omega(n \polylog n)$), it is \emph{not} time-optimal; the time complexity
for spanning tree and leader election is $\tilde{O}(n)$.
The works of \cite{GhaffariK18,gmyr} showed bounds with 
with improved round complexity, but with worse bounds on the message complexity  and more generally, trade-offs between time and messages \cite{gmyr}. We note that all these
results are for \emph{synchronous} networks. For \emph{asynchronous} networks, Mashreghi and King~\cite{KMDISC19,MK21} presented a spanning tree algorithm (also applies for leader election and MST) that
takes $\tilde{O}(n^{1.5})$ messages and $\tilde{O}(n)$ time. 
It is an open question whether one can design a randomized (non-comparison based) algorithm that takes $\tilde{O}(n)$ messages and $\tilde{O}(D)$ rounds in the $KT_1$ model; these
are the optimal bounds possible in the $KT_1$ model.

\section{A Singularly Optimal Asynchronous Leader Election Algorithm}
\label{sec:alg-analysis}
In this section, we present a leader election algorithm in asynchronous networks that is essentially optimal with respect to both message and time
complexity. We assume that all nodes have knowledge of $n$, the number of nodes in the network.

\subsection{Algorithm}

\noindent \textbf{Brief Overview.} The process is initiated by one or more nodes, which are woken up by the adversary, 
possibly at different times (see earlier discussion of the adversarial wakeup model in Section~\ref{sec:intro}).
These nodes wake up the rest of the nodes.
Subsequently, each awake node decides whether it will participate in the 
algorithm in the role of (i) a candidate and/or (ii) a referee.
To do that, the node chooses randomly (for each role separately), 
with probability $O(\log n/ n)$, whether to take on the role or not.\footnote{It is possible for a node to participate in both roles or not participate in either.} Each {\em candidate} attempts to become the leader by winning over a sufficient number of referees. {\em Referees} are used to decide which candidate will go on to become the leader, essentially preferring a stronger candidate
(one who randomly chose a higher rank) over a weaker (lower rank) one, provided the stronger candidate does not ``show up too late'' (namely, after the weaker candidate has already accumulated sufficiently many referees to declare itself leader).
Once a candidate becomes a leader, it broadcasts the message that it is leader and that all nodes should terminate.

\noindent \textbf{Detailed Description.}
Each awake node $u$ maintains the following information: (i) a rank $\ID_u$, chosen uniformly at random from $[1,n^4]$ (we assume that all nodes have knowledge of $n$, the network size), (ii) two indicator variables $\CANDSTATE$ and $\REFSTATE$ reflecting the status of its candidacy for leadership and its role as a referee, respectively, and (iii) the set $\MLIST$ of all messages $u$ has heard over the course of the algorithm.\footnote{Actually, our algorithm can be extended if nodes have a knowledge of $n$ that is within some (known) constant factor. We discuss this 
at the end of Section~\ref{subsec:analysis}.} The variable $\CANDSTATE$ can take one of three values, $\candidateSTATE$, $\nonelectedSTATE$, and $\electedSTATE$, denoting whether the node is a candidate for leadership, it irrevocably committed itself to not be a leader, or it irrevocably committed itself to be a leader, respectively. 
The variable $\REFSTATE$ can take one of four values $\RefNonselectedSTATE$, $\RefReadySTATE$, $\RefChosenSelectedSTATE$, and $\RefInDisputeSTATE$, with each of the states described in detail later on.

Initially, a node $u$ is asleep and may be awoken by either  
the adversary 
or a $\langle \WAKEUP \rangle$ message that originated at another node.
(For the sake of uniformity, it is convenient to think of the 
the adversary waking up a node 
as a $\WAKEUP$ message arriving from outside the system, and treat both events in the same way.)
Once awoken, $u$ calls Procedure $\Initialize$, which first wakes up adjacent nodes by sending the message $\langle \WAKEUP \rangle$ to $u$'s neighbors.\footnote{Here, and in other places, we say that a node $u$ sends messages only to its neighbors. This is the local description of the algorithm. Globally, this results in the message being broadcast throughout the graph.} Subsequently, $u$ chooses randomly whether or not to be a candidate by flipping a biased coin with probability $1000 \log n/n$: if successful, $u$ initializes $\CANDSTATE$ to $\candidateSTATE$, chooses a rank at random - an integer in $[1,n^4]$, and broadcasts a request carrying its rank.\footnote{The theorems in this paper hold for all sufficiently large $n$, say $n\ge n_0$, where the value of $n_0$ depends on the probability used for the coin flip. The particular value of $1000 \log n/n$ was chosen for ease of exposition, but it can be modified in order to yield a smaller $n_0$. No attempt was made to optimize this value.}  
If not successful, $u$  sets $\CANDSTATE$ to $\nonelectedSTATE$. Node $u$ also chooses randomly whether or not to be a referee, with probability $1000 \log n/n$: if successful, $u$ initializes $\REFSTATE$ to $\RefReadySTATE$ and sets the variables $\Contender$ and $\Chosen$ to $-1$,  else it sets $\REFSTATE$ to $\RefNonselectedSTATE$. 
The initialization procedure is described formally in Algorithm~\ref{alg:initialization}.

Consider some node $v$. Any message that $v$ wants to send, whether to continue a previous broadcast or to pass on a newly generated message, is added to its list $\SENDLIST(e)$ for each of its edges $e$. 
Whenever one of $v$'s outgoing edges $e$ is free for a new message to be sent across it, $v$ invokes Procedure $\SendMessage$ for that edge, which picks an arbitrary unsent message from the list $\SENDLIST(e)$, transmits it over the edge, and erases it from the list.
Any messages that were generated by the node are 
added also to $\MLIST$, a list of messages the node has already heard. This process is described formally in Algorithm~\ref{alg:send-message}.

If node $u$ becomes a candidate (respectively, a referee), then its subsequent actions in this role are governed by Procedure $\Candidate$ (resp., 
Procedure $\Referee$).\footnote{Recall that these two procedures may run concurrently in the same node, in case it assumed both roles.} 
The appropriate procedure between these two procedures is called when the node receives a message. Another node's leader announcement message is processed outside of these two procedures.
Additionally, when a relevant message is received, the Procedure $\CandidateDecideResponse$ may also be invoked to resolve a dispute (to be described later). In addition to whichever procedure is called, if any, every awake node stores all the messages it receives in its list $\MLIST$ and participates in every broadcast that reaches it by also adding the message to $\SENDLIST(e)$ for all its edges $e$ except those over which the message was received.\footnote{Note that if a message currently in $v'$s $\SENDLIST(e)$ is received over edge $e$, then $v$ removes that message from $\SENDLIST(e)$.}
A node $u$, once awake, will eventually either broadcast a $\LEADER$ message (in its role as a candidate) and terminate, or receive such a message about another candidate, at which point $u$ stores the rank of the leader and then terminates
(having also forwarded the transmitted message to its neighbors).
The process governing the responses of the nodes to different incoming messages is described in Algorithm~\ref{alg:on-receive-message} (called Procedure $\OnReceiveMessage$). 

Procedure $\Candidate$ is run by a candidate to help it determine whether it will become the leader or not. Each candidate $u$ broadcasts a request message during Procedure $\Initialize$. Candidate $u$ then waits for $900 \log n$ replies from different referees.\footnote{A candidate may receive more than $900 \log n$ replies as there may be more referees. In such a case, the candidate only considers the \textit{first} $900 \log n$ replies it receives.} If one of these replies is a decline message, i.e., a referee says that $u$ is declined from being the leader, then $u$ updates its $\CANDSTATE$ to $\nonelectedSTATE$ and sends the message $\langle \ID_u, \LOSES\rangle$ to all neighbors. 
Otherwise if all the replies are $\APPROVED$ messages (and $u$ has not terminated yet), $u$ becomes the leader, i.e., it changes its $\CANDSTATE$ to $\electedSTATE$, then announces this to all the nodes and terminates.\footnote{Note that messages about another candidate becoming the leader (and thus possibly affecting $u$'s candidacy) are not handled in this procedure but are instead handled in Procedure $\OnReceiveMessage$.} 
The state progression of a candidate is pictorially represented in Figure~\ref{fig:cand-state-progression}. The pseudocode for the candidate's actions is given in Algorithm~\ref{alg:candidate}.
\begin{figure}
    \centering
    \includegraphics[height=2in]{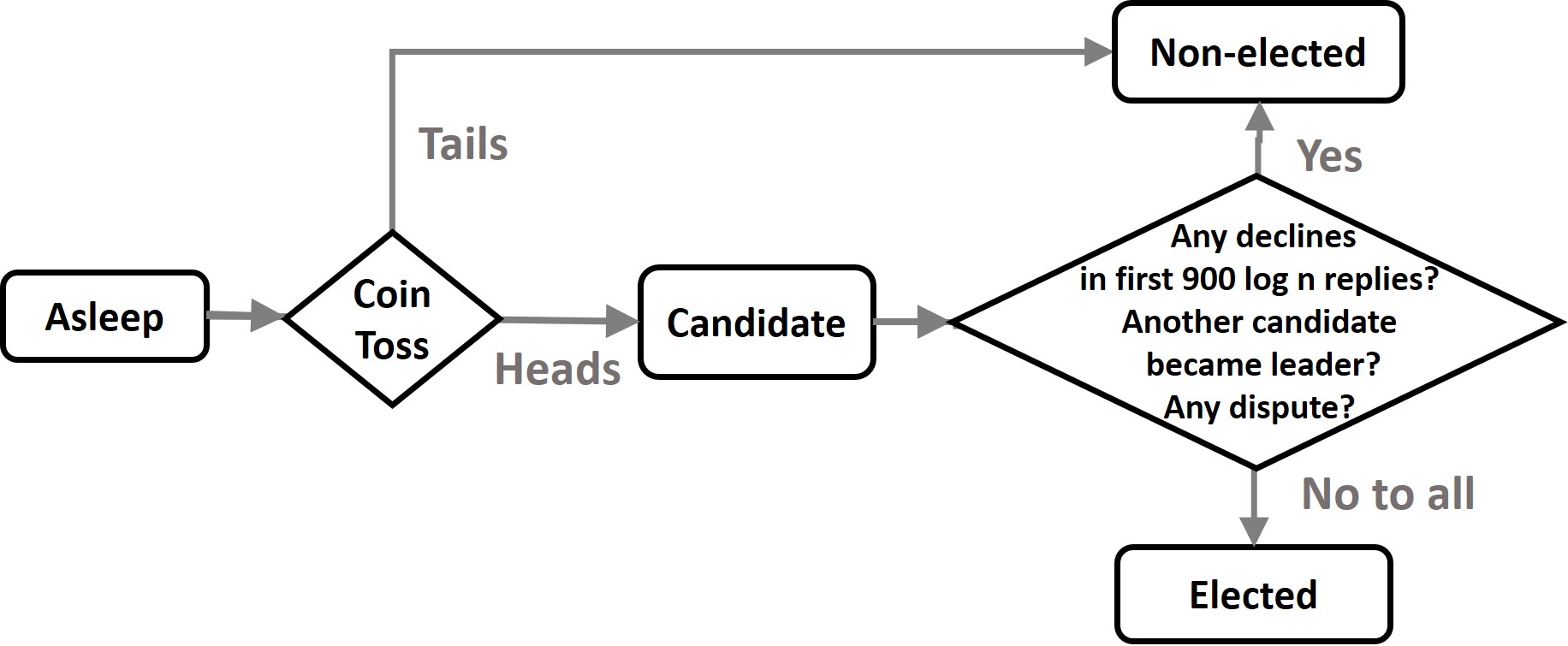}
    \caption{The state progression of a candidate.}
    \label{fig:cand-state-progression}
\end{figure}

Procedure $\Referee$ is run by a referee $r$ to decide which candidate becomes the leader.
A referee's $\REFSTATE$ may be set to one of the following three values.
\begin{itemize}
\item
$\REFSTATE= \RefReadySTATE$ holds when $r$ has not heard yet of any candidate. 
\item
$\REFSTATE= \RefChosenSelectedSTATE$ holds when $r$ has approved one candidate, referred to as its {\em chosen} candidate, and has declined every other candidate that approached it so far.
The rank of the chosen candidate is stored in the variable $\Chosen$.
\item
$\REFSTATE=\RefInDisputeSTATE$ holds when $r$ currently keeps track of two candidates:
the {\em chosen} $v$, whose rank $\ID_v$ is stored in the variable $\Chosen$, and a {\em contender} $w$, whose rank $\ID_w$ is stored in the variable $\Contender$, such that $\ID_w > \ID_v$, and all other candidates that approached $r$ so far were declined. In such a case, a {\em dispute} is currently in progress between $v$ and $w$. 
\end{itemize}

The state $\RefInDisputeSTATE$ is typically reached when $r$ has a chosen candidate $v$ that was approved by it, and later it gets a request from another candidate $w$ such that $\ID_w > \ID_v$. In this situation, $r$ cannot decline $w$ (since it has a higher $\ID$ than its current chosen), but at the same time,
it cannot approve $w$, since it may be that $v$ had already collected enough approvals and became the leader in the meantime. To resolve this uncertainty, $r$ declares a dispute, and broadcasts a  
message containing the dispute information, $\langle \ID_v, \ID_w, \DECIDE \rangle$, so that when $v$ eventually receives the message, it can make the comparison between itself and $w$ and resolve the dispute 
($v$ wins iff it already became the leader prior to receiving the dispute message).
While waiting for $v$'s response, $r$ stores $\ID_w$ of $w$ in the variable $\Contender$. 

As can be seen, the referee changes state only as a consequence of receiving a request from a candidate or receiving the result of a dispute from its chosen candidate. 
The pseudocode for the referee's actions is given in Algorithm~\ref{alg:referee}.

\begin{sloppypar}
We now observe how a referee $r$ responds upon receiving a message  $\langle \ID_u, \REQUEST\rangle$ from a candidate $u$. We only consider what happens if $r$'s $\REFSTATE \neq \RefNonselectedSTATE$ as otherwise $r$ is not a referee. Three states are possible.
\end{sloppypar}

\noindent
(A) 
If $r$'s $\REFSTATE = \RefReadySTATE$, meaning that $u$ is the first contender approaching $r$ with a request, then $r$ registers $u$ as its chosen candidate by storing $\ID_u$ in the variable $\Chosen$, broadcasts an approval message for $u$, $\langle \ID_u, \ID_r, \APPROVED \rangle$, and switches its $\REFSTATE$ to $\RefChosenSelectedSTATE$.

\noindent
(B)
If $r$'s $\REFSTATE = \RefChosenSelectedSTATE$, and the current chosen candidate is $v$, then there are two possibilities. The simpler situation is when $v$ is stronger than $u$, i.e., $\ID_v > \ID_u$, in which case $r$ may immediately broadcast a decline message for $u$, $\langle \ID_u, \ID_r, \DECLINED \rangle$.

The other case is that $u$ is the stronger of the two candidates, i.e., $\ID_v < \ID_u$. In this case, $u$ should normally replace $v$ as the chosen candidate, except if $v$ has already declared itself leader in the meantime. The way to resolve this question is a dispute between $u$ and $v$. First, $r$ checks its list $\MLIST$ of previously received messages to see if such a dispute between $u$ and $v$ is already in progress, i.e., if $\MLIST$ contains a previously received message of the form $\langle \ID_v, \ID_u, \DECIDE\rangle$ announcing the initiation of a dispute between $u$ and $v$, or even a message of the form $\langle \ID_v, \LOSES \rangle$ announcing the outcome of such a dispute (such a message necessarily indicates that $v$ has {\em lost} the dispute, since a ``win'' by $v$ can only occur if $v$ has already declared itself leader, in which case $v$ has already broadcast this fact and hence need not reply to the dispute).\footnote{Note that a  $\langle \ID_v, \LOSES \rangle$ message may be the result of a dispute between a candidate $v$ and some other candidate, not necessarily $u$. It is sufficient that $v$ lost its candidacy.} 
 There are three possible situations.
\begin{itemize}
\item
Referee $r$ has already received a message with the outcome of the dispute (namely, $v$ lost). 
Then $r$ only updates $\Chosen$ to $\ID_u$ and broadcasts an approval message for $u$, 
$\langle \ID_u, \ID_r, \APPROVED \rangle$.
\item
Referee $r$ has received a 
 message announcing a dispute, but has not yet heard about the outcome. 
Then $r$ only updates $\Contender$ to $\ID_u$ and updates $\REFSTATE$ to $\RefInDisputeSTATE$ and awaits news on the outcome (but does not broadcast any new messages). 
\item 
Referee $r$ did not hear of an existing dispute between $u$ and $v$.
Then it is up to $r$ to initiate a dispute, so $r$ registers $u$ as the contender by storing $\ID_u$ in the variable $\Contender$,
sets its $\REFSTATE$ to $\RefInDisputeSTATE$, and broadcasts a 
 dispute message for $v$ of the form $\langle \ID_v, \ID_u, \DECIDE \rangle$. 
\end{itemize}

\begin{sloppypar}
\noindent
(C) 
If $r$'s $\REFSTATE = \RefInDisputeSTATE$, signifying that {\em another} dispute (between $v$ and some other candidate) is in progress, then $r$ compares the new candidate $u$ with the current contender $w$.
If $u$ is weaker than $w$ ($\ID_u < \ID_w$), then $r$ immediately broadcasts a decline message for $u$, $\langle \ID_u, \ID_r, \DECLINED \rangle$. 
Otherwise ($\ID_u > \ID_w$), $r$ broadcasts a decline message for $w$, $\langle \ID_w, \ID_r, \DECLINED \rangle$, updates $\Contender$ to $\ID_u$, and initiates a new dispute by broadcasting a 
 dispute message for $\Chosen$ of the form $\langle \Chosen, \ID_u, \DECIDE \rangle$. 
\end{sloppypar}
The pseudocode for the referee's actions on receiving a request is given in Algorithm~\ref{alg:referee-request-response} (called Procedure $\RefereeRequestResponse$). 

Finally, let us describe how a node $v$ which is currently the chosen candidate of some referee $r$ (but may have possibly changed its $\CANDSTATE$ since the time it was approved by $r$) handles a 
 dispute request. 
Notice that a $\langle \ID_v, \ID_u, \DECIDE\rangle$ message is only sent to $v$ when it is weaker than $u$ ($\ID_v < \ID_u$), so when $v$ receives such a message, it must 
abdicate its candidacy (if it is still a candidate), unless it has already elected itself as leader (which is an irreversible decision) and terminated. 
Hence, if $v$'s 
$\CANDSTATE = \candidateSTATE$, then $v$ relinquishes its candidacy by setting $\CANDSTATE$ to $\nonelectedSTATE$ and broadcasts the result of the dispute as 
$\langle \ID_v, \LOSES \rangle$.\footnote{Note that this broadcast operation is 
unnecessary when $v$'s $\CANDSTATE = \nonelectedSTATE$, as $v$ would have previously broadcast a message announcing its loss. This broadcast would have been the result of either another dispute involving $v$ or $v$ receiving a decline from one of the referees.}  If, however, $v$'s $\CANDSTATE$ is set to $\electedSTATE$,
then $v$ has already broadcast a leader announcement message and terminated, so no additional response to the 
dispute message is required. Pseudocode for the actions of the chosen candidate on receiving a 
dispute message from a referee is given in Algorithm~\ref{alg:candidate-decide-response} (called Procedure $\CandidateDecideResponse$).

\begin{sloppypar}
As there are multiple referees generating 
 messages for various disputes, a referee may receive the results of a dispute it does not need to immediately process (but the result is stored for future processing, if any). 
If $r$'s $\REFSTATE = \RefInDisputeSTATE$, $r$'s chosen is $v$, $r$'s contender is $u$, and $r$ receives a reply to a dispute of the form 
$\langle \ID_v, \LOSES \rangle$, then $r$ immediately processes the message as follows.\footnote{Note that the message $\langle \ID_v, \LOSES \rangle$ may also be generated by a candidate $v$ upon receiving a decline message from a referee. For ease of writing, when we say ``reply to a dispute'', we mean any message of the form $\langle \ID_v, \LOSES \rangle$, regardless of how it was generated.}
First, $r$ updates $\Chosen$ to $\ID_u$, sets $\Contender$ to $-1$, and updates $\REFSTATE$ to $\RefChosenSelectedSTATE$. Subsequently, $r$ initiates the broadcast of an approval message for $u$, $\langle \ID_u, \ID_r, \APPROVED \rangle$. The referee's actions on receiving a reply from the chosen about an ongoing dispute are given in Algorithm~\ref{alg:referee-dispute-response} (called Procedure $\RefereeDisputeResponse$). The state progression of a referee is seen in Figure~\ref{fig:ref-state-progression}.
\end{sloppypar}

\begin{figure}
    \centering
    \includegraphics[page=7,height=2in]{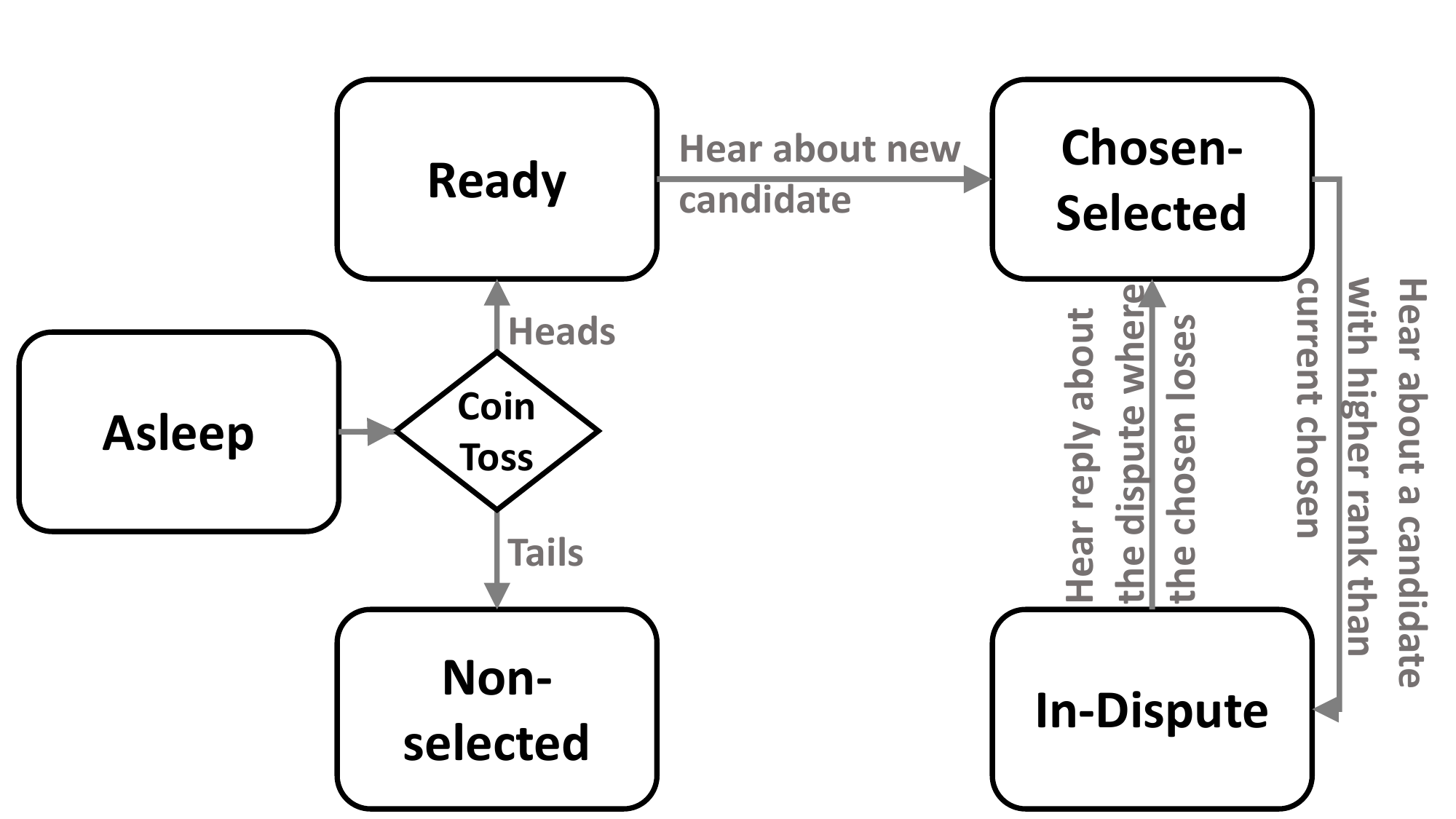}
    \caption{The state progression of a referee.}
    \label{fig:ref-state-progression}
\end{figure}

\subsection{Analysis}
\label{subsec:analysis}

We now prove the correctness of the algorithm and analyze its complexity. We prove a weaker time complexity bound of $O(D \log^2 n)$ here, with the stronger result of $O(D + \log^2 n)$ in Section~\ref{subsubsec:reduced-time-complexity}.
Before getting into the meat of the proof, let us make an important observation and subsequently state a useful lemma.

\begin{observation}\label{obs:everyone-wakesup}
From the time the first node is awake, all nodes are awakened in $O(D)$ time using $O(m)$ messages.
\end{observation}

Set $\REFLOW = 900 \log n$ and $\REFHIGH = 1100 \log n$. By Observation~\ref{obs:everyone-wakesup}, we see that all nodes awaken and thus participate in candidate selection and referee selection. 
Denote by $N_C$ and $N_R$ the number of candidates and referees selected in the algorithm, respectively.
We now bound $N_C$ and $N_R$ with high probability.

\begin{lemma}\label{lem:num-cand}\label{lem:num-ref}
With probability $1 - 1/n^3$, both the number of candidates $N_C$ and the number of referees $N_R$ are in $[\REFLOW, \REFHIGH]$.
\end{lemma}

\begin{proof}
The random choice of a single candidate can be viewed as a Bernoulli trial with probability $1000 \log n/n$. Thus, the total number of candidates chosen $N_C$ is the sum of $n$ independent Bernoulli trials. Using known Chernoff bounds such as Theorem 4.4 and 4.5 in~\cite{MU17} where $\delta = 1/10$ and $\mu = 1000 \log n$, we see that the bounds in the lemma hold with the required probability. 
A similar argument holds for $N_R$.
\end{proof}

We are now ready to argue the correctness of the algorithm, i.e., show that exactly one candidate becomes a leader with high probability. This is done in two stages, showing first that the number of candidates that become leaders is {\em at least} one and then that this number is {\em at most} one. Throughout the following lemmas, we require that each node has a unique $\ID$. It is easy to see that this is true with high probability since each node selects its rank uniformly at random from $[1,n^4]$.

\begin{lemma}\label{lem:at-least-one-cand}
At least one candidate becomes a leader with high probability.
\end{lemma}

\begin{proof}
By Lemma~\ref{lem:num-cand}, at least one node becomes a candidate with high probability and chooses a \ID. Each candidate waits for replies from $\REFLOW$ referees before deciding to become a leader. By Lemma~\ref{lem:num-ref}, at least that many nodes become referees with high probability. Thus, there exist enough referees with high probability that generate replies for a candidate so that it can become a leader. 

Let $u$ be the candidate with the highest \ID. Now $u$ broadcasts its candidacy to all referees. Either $u$ wins at every referee, receives the responses, and becomes a leader. Or else, $u$  loses at some referee, which implies that some other candidate is a leader. Thus at least one candidate becomes a leader.
\end{proof}

\begin{lemma}\label{lem:at-most-one-cand}
At most one candidate becomes a leader with high probability.
\end{lemma}

\begin{proof}
Each candidate waits for positive replies from $\REFLOW$ referees in order to become a leader. By Lemma~\ref{lem:num-ref}, with high probability, the number of referees that exist in the system satisfies 
$$N_R\ge \REFLOW > 0.8 N_R.$$
Put another way, given that $N_R\in [\REFLOW, \REFHIGH]$, for any two candidates $u$ and $v$ that receive replies from the sets of referees $R_u$ and $R_v$ in order to decide on becoming a leader,
\begin{align*}
|R_u \cap R_v| &=~ |R_u \cup R_v| - |R_u \setminus R_v| - |R_v \setminus R_u| 
~\geq~ 
\REFLOW - (\REFHIGH - \REFLOW) - (\REFHIGH - \REFLOW) \\
 &=~ 900 \log n - (1100 \log n - 900 \log n) - (1100 \log n - 900 \log n) 
 ~ > ~ 1.
\end{align*}

We show that when two candidates $u$ and $v$ share a referee $r$ whose replies help them determine if they may become a leader, 
it is impossible for both candidates to become leaders. Thus, no more than one candidate becomes a leader with high probability. Without loss of generality, let $\ID_u < \ID_v$. Consider the sequence of arrival of candidacy messages at $r$ and replies. The following three cases cover all possible scenarios.

\inline Case 1:
$r$ knows of a candidate $w$ where either $\ID_w > \ID_u$ or $w$ has become a leader before a dispute message generated by $r$ reaches it. 

If $r$ knows of a candidate $w$ where $\ID_w > \ID_u$, then it is clear that at least $u$ will be rejected and not become the leader. Otherwise, if $r$ receives a candidacy message from either $u$ or $v$, $r$ will generate a dispute message and send it to $w$. If $w$ has become the leader before this dispute message reaches it, then $w$ would have already generated a leader announcement message and terminated. Node $r$ meanwhile will not confirm $u$ or $v$ as a leader until it hears back from $w$. So whichever of $u$'s or $v$'s candidacy message was at $r$ will not be approved once $w$'s leader announcement message reaches node $r$ and both candidates will not become the leader. 

In this situation, it is guaranteed that at least one of $u$ or $v$ not become the leader. 

\inline Case 2:
$v$'s candidacy message reaches $r$ first and $r$ subsequently replies that $v$ may become the leader, all before $u$'s candidacy message reaches $r$. 

In this case, since $\ID_u < \ID_v$, $r$ declines $u$ once its candidacy message reaches $r$. It is also possible that $v$ may have become the leader and generated a leader announcement message. In this situation, that message may propagate to $r$ and $u$, also resulting in $u$ not becoming the leader. In either situation, $u$ will not become the leader.

\inline Case 3:
$u$'s candidacy message reaches $r$ first and $r$ subsequently replies that $u$ may become the leader, all before $v$'s candidacy message reaches $r$. 

In this case, if $u$ received enough approvals and became the leader, then it may generate a leader announcement message. Now, this leader announcement message will either reach $r$ before or after $v$'s candidacy message reaches it. In either case, $v$ will not become the leader because $r$ will not approve $v$ before receiving the result of the dispute from $u$. If however, $u$ did not receive enough approvals before $v$'s candidacy message reaches $r$ and $r$'s subsequently generated dispute reaches $u$, then $u$ will give up its candidacy.

In either case, at most one of $u$ or $v$ will become the leader (perhaps neither of them).
\end{proof}

Thus, the algorithm is correct. We now give a useful lemma that is subsequently used to bound both the message and time complexity of the algorithm.


\begin{lemma}\label{lem:num-unique-msg-per-node}
Any node generates at most $O(\log n)$ unique messages with high probability to be broadcast over the course of the algorithm.
\end{lemma}

\begin{proof}
There is one instance of a $\WAKEUP$ message sent through the system. In addition, each candidate generates one candidacy request message and possibly one leader announcement message or candidacy loss message (as a reply to a dispute). Thus each candidate generates $O(1)$ messages. 
Each referee generates a reply to each candidate it hears from and possibly a dispute message with the current chosen as well. By Lemma~\ref{lem:num-cand}, there are $O(\log n)$ candidates a referee may have to reply to and generate disputes for, resulting in each referee generating $O(\log n)$ messages.
\end{proof}

\begin{lemma}\label{lem:total-unique-msgs}
There are $O(\log^2 n)$ unique messages with high probability broadcast in the system over the course of the algorithm.
\end{lemma}

\begin{proof}
Aside from the initial $\WAKEUP$ message, only candidates and referees generate unique messages to be broadcast. There are $O(\log n)$ such candidates and referees with high probability by Lemma~\ref{lem:num-cand}. Thus, there are totally $O(\log^2 n)$ unique messages with high probability broadcast in the system over the course of the algorithm.
\end{proof}

Combining the lemma with the fact that each broadcast of a unique message results in $O(m)$ messages, we get the following.
\begin{corollary}
The total message complexity is $O(m \log^2 n)$ with high probability.
\end{corollary}


\begin{lemma}
The run time of the algorithm is $O(D \log^2 n)$ with high probability.
\end{lemma}

\begin{proof}
By Observation~\ref{obs:everyone-wakesup}, all nodes wake up in $O(D)$ time. 
One of the woken nodes,
say $u$, will go on to become the leader. Let us bound the number of ``logical phases'' of broadcasts needed until all the nodes are made aware that $u$ is the leader.
(Note that we only use the word ``phase'' in this analysis, but we do not use this terminology in the algorithm itself.) 
Each phase is responsible for certain information being broadcast to nodes, and the phases, as described, occur sequentially. Furthermore, different phases may take different amounts of time. In the first phase, $u$ broadcasts its candidacy. In the next phase, each of the referees may need to broadcast a dispute. In the subsequent phase, each of the candidates that is a target of a dispute needs to broadcast its reply. In the next phase, each of these referees broadcasts its reply to $u$'s candidacy request. In the final phase, $u$ broadcasts that it is the leader. Thus, there are $O(1)$ such phases. 

In each of these phases, a broadcast originating at some node $u$ is complete when the message $m$ reaches all other nodes. The shortest path between $u$ and any other node is of length at most $O(D)$. By Lemma~\ref{lem:total-unique-msgs}, there are at most $O(\log^2 n)$ unique messages generated in the system with high probability. Thus $m$ may be delayed at each node in the shortest path by at most $O(\log^2 n)$ other messages with high probability, resulting in a total time of $O(D \log^2 n)$ with high probability for the phase to complete. Since there are $O(1)$ phases, the total time until the algorithm completes is $O(D \log^2 n)$ with high probability.
\end{proof}

\begin{theorem}
There exists an algorithm that solves 
leader election with high probability in any arbitrary graph with $n$ nodes, $m$ edges, and diameter $D$ in $O(D \log^2 n)$ time with high probability using $O(m \log^2 n)$ messages with high probability in an asynchronous system with adversarial node wakeup.
\end{theorem}

\subsection{Improvements}

\subsubsection{Knowledge of $n$}
In the above analysis, it may be noted that nodes do not need to know the exact value of $n$. In fact, it is easy to extend the algorithm and analysis
if nodes know the value of $n$ up to a  constant factor. More precisely, it is sufficient if all nodes know either (i) the value of a constant $c_1$, where $0 < c_1 \leq 1$, and a lower bound on $n$, $n'$, such that $c_1 n \leq n' \leq n$ or (ii) the value of a  constant $c_2$, where $1 \leq c_2$, and an upper bound on $n$, $n^*$, such that $n \leq n^* \leq c_2 n$. 
By adjusting the coin toss probability to some $c \log n'/n'$ (or $c \log n^* / n^*$) for a carefully chosen value of $c$, we can show that the analysis goes through for a sufficiently large $n$.

\subsubsection{Reducing Time Complexity to $O(D + \log^2 n)$}
\label{subsubsec:reduced-time-complexity}
The analysis of the runtime in Section~\ref{subsec:analysis} can be tightened further. The below lemma comes from Theorem~1 in Topkis~\cite{T89}, adapted to the current setting and terminology. Notice that we can use their Theorem because the process of flooding they study is being implemented here via $\SendMessage(e)$.

\begin{lemma}\label{lem:reduced-time-complexity}
It takes $D+k-1$ time to broadcast $k$ messages in a graph with diameter $D$ in the asynchronous setting.
\end{lemma}

Combining Lemma~\ref{lem:reduced-time-complexity} with Lemma~\ref{lem:total-unique-msgs}, which states that there are at most $O(\log^2 n)$ unique messages with high probability, and the argument (from the proof in Section~\ref{subsec:analysis}) that there are $O(1)$ ``logical phases'' of broadcasts, we see that the run time of the algorithm is $O(D + \log^2 n)$ with high probability. When coupled with our previous analysis of correctness and message complexity, we arrive at the following theorem.
\begin{theorem}
There exists an algorithm that solves leader election with high probability in any arbitrary graph with $n$ nodes, $m$ edges, and diameter $D$ in $O(D + \log^2 n)$ time with high probability using $O(m \log^2 n)$ messages with high probability in an asynchronous system with adversarial node wakeup.
\end{theorem}

\section{Conclusion}\label{sec:conclusion}
We have presented a randomized algorithm for asynchronous leader election in general networks. Our algorithm has message and time bounds
that are both within a polylogarithmic factor of the lower bound, and is the first such singularly optimal algorithm
presented for general asynchronous networks. 

Two important open questions remain. First, is it possible to obtain near singularly optimal bounds
 using a \emph{deterministic} algorithm?
Our algorithm needs an accurate knowledge of the network size (at least up to a constant factor). 
Second, can we get a (near) singularly optimal algorithm (even randomized) without the restriction
that nodes need an accurate knowledge of $n$ or even no knowledge of $n$? 
Third, can we design (near) singularly optimal algorithms for other fundamental
problems such as minimum spanning tree and shortest paths in the asynchronous model.

\bibliographystyle{plainurl}
\bibliography{references-all}

\clearpage

\appendix

\centerline{\LARGE \textbf{Appendix}}
\section{Formal code}\label{sec:code}

\alglanguage{pseudocode}
\begin{algorithm}[H]
	\caption{Procedure $\Initialize(u)$, run by node $u$.}
	\label{alg:initialization}
	
	\begin{algorithmic}[1]
	\State Choose an integer in $[1,n^4]$ uniformly at random to be $u$'s rank $\ID_u$
	\State Flip a biased coin $C_C$ with probability $1000 \log n/n$
    \If{$C_C$ comes heads}
		\State $\CANDSTATE \gets \candidateSTATE$ 
		\State $\NUMREPLIES \gets 0$
		\State Add $\langle \ID_u,\REQUEST \rangle$ to $\SENDLIST(e)$ for every edge $e$ adjacent to $u$ 
	\Else 
		\State $\CANDSTATE \gets \nonelectedSTATE$
	\EndIf

	\State Flip another biased coin $C_R$ with probability $1000 \log n/n$
    \If{$C_R$ comes heads}	
	    \State $\REFSTATE \gets \RefReadySTATE$
	    \State Set $\Contender \gets -1$
	    \State Set $\Chosen \gets -1$
	\Else
	    \State $\REFSTATE \gets \RefNonselectedSTATE$
	\EndIf

	\end{algorithmic}
\end{algorithm}

\alglanguage{pseudocode}
\begin{algorithm}[H]
	\caption{Procedure $\SendMessage(u,e)$, run by node $u$ whenever it can send a new message on a given edge $e$.}
	\label{alg:send-message}
	
	\begin{algorithmic}[1]
	\If{$\SENDLIST(e)$ is not empty (* $u$ has a message to be sent on edge $e$ *)}
	    \State Choose any message $M$ in $\SENDLIST(e)$
	    \State Send $M$ on $e$
	    \State Add $M$ to $\MLIST$ if not already present there (* $u$ might have generated this message and thus not heard it from another node *)
	\EndIf
	\end{algorithmic}
\end{algorithm}

\alglanguage{pseudocode}
\begin{algorithm}[H]
	\caption{Procedure $\OnReceiveMessage(u,M)$, run by node $u$ upon receiving a message $M$.}
	\label{alg:on-receive-message}
	
	\begin{algorithmic}[1]
	\If{the message $M$ is in $\MLIST$ (* $u$ previously heard of message $M$ *)}
	    \State If $M$ was received over edge $e$, remove $M$, if present, from $\SENDLIST(e)$
	\Else{  (* $u$ has not previously heard of message $M$ *)}
	    \State Add $M$ to $\MLIST$
	    \State Add $M$ to $\SENDLIST(e)$ for each edge $e$ of $u$ except the edge the message was received on
	    \If{$M = \langle \WAKEUP \rangle$}
	        \State Invoke Procedure $\Initialize$
	    \ElsIf{$M = \langle \ID_v, \LEADER \rangle$ (* $u$ receives a message about $v$ being the leader *)}
	        \State Set leader as $v$ (* Note that necessarily $v\ne u$ *)
	        \State Set $\CANDSTATE \gets \nonelectedSTATE$
	        \State Terminate
	    \ElsIf{$M = \langle \ID_u, \ID_r, \DECLINED \rangle$ OR $M = \langle \ID_u, \ID_r, \APPROVED \rangle$  (* $u$ receives a response from a referee $r$ about candidacy *)}
	        \State Invoke Procedure $\Candidate(u)$
	   \ElsIf{$M = \langle \ID_u, \ID_v, \DECIDE \rangle$ (* $u$ receives a  
	   dispute message about another candidate $v$ *)}
	        \State Invoke Procedure $\CandidateDecideResponse(u)$
	   \ElsIf{$\REFSTATE \neq \RefNonselectedSTATE$}
	        \State Invoke Procedure $\Referee(u)$
	   \EndIf
	\EndIf
	\end{algorithmic}
\end{algorithm}

\alglanguage{pseudocode}
\begin{algorithm}[H]
	\caption{Procedure $\Candidate(u,M)$, run by candidate $u$ on receiving a reply $M$ regarding its candidacy.}
	\label{alg:candidate}
	
	\begin{algorithmic}[1]
	\If{$\CANDSTATE = \Candidate$}
	    \If{$M = \langle \ID_u,\ID_r, \DECLINED\rangle$ (* Where $r$ is presumably a referee *)}
	        \State $\CANDSTATE \gets \nonelectedSTATE$
            \State Add $\langle \ID_u, \LOSES\rangle$ to $\SENDLIST(e)$ for every edge $e$ adjacent to $u$  
	    \ElsIf{$M = \langle \ID_u, \ID_r, \APPROVED \rangle$ (* Where $r$ is presumably a referee *)}
	        \State $\NUMREPLIES \gets \NUMREPLIES + 1$
	        \If {$\NUMREPLIES = 900 \log n$ (* $u$ has received approvals from $900 \log n$ referees *)}
	            \State $\CANDSTATE \gets \electedSTATE$
                \State Add $\langle \ID_u,\LEADER\rangle$ to $\SENDLIST(e)$ for every edge $e$ adjacent to $u$ 
	            \State Terminate
	        \EndIf
	    \EndIf
	\EndIf
	\end{algorithmic}
\end{algorithm}

\alglanguage{pseudocode}
\begin{algorithm}[H]
	\caption{Procedure $\Referee(r,M)$, run by referee $r$ upon receiving a message $M$.}
	\label{alg:referee}
	
	\begin{algorithmic}[1]
    \If {$M = \langle \ID_u, \REQUEST \rangle$ (* $r$ receives a candidacy request from a candidate $u$ *)}
        \State Invoke Procedure $\RefereeRequestResponse(r)$
    \ElsIf {$M = \langle \ID_v, \LOSES \rangle$ (* $r$ receives the result of a dispute between $v$ and another candidate *)}
        \If{$\REFSTATE= \RefInDisputeSTATE$ AND $\Chosen = \ID_v$ 
        (* $r$ was waiting for such a message *)}
        \State Invoke Procedure $\RefereeDisputeResponse(r)$
        \EndIf
    \EndIf
	
	\end{algorithmic}
\end{algorithm}

\alglanguage{pseudocode}
\begin{algorithm}[H]
	\caption{Procedure $\CandidateDecideResponse(v,M)$, run by a node $v$ upon receiving the message $M=\langle \ID_v, \ID_u, \DECIDE\rangle $.}
	\label{alg:candidate-decide-response}
	
	\begin{algorithmic}[1]

    \If{$\CANDSTATE = \candidateSTATE$}
	    \State Set $\CANDSTATE \gets \nonelectedSTATE$
        \State Add $\langle \ID_v, \LOSES\rangle$ to $\SENDLIST(e)$ for every edge $e$ adjacent to $v$ 
	\EndIf
     (* Note that $v$ {\em wins} the dispute with a stronger candidate $u$ only if it has already declared itself leader, hence in this case no further response is necessary. *)
	\end{algorithmic}
\end{algorithm}

\alglanguage{pseudocode}
\begin{algorithm}[H]
	\caption{Procedure $\RefereeRequestResponse(r,M)$, run by referee $r$ on receiving a message $M=\langle \ID_u, \REQUEST\rangle$ originating from candidate $u$.}
	\label{alg:referee-request-response}

\begin{algorithmic}[1]
\If{$\REFSTATE=\RefReadySTATE$ (* $u$ is the first candidate to approach $r$ *)}
	\State Set $\Chosen\gets \ID_u$
    \State Add $\langle \ID_u, \ID_r, \APPROVED \rangle$ to $\SENDLIST(e)$ for every edge $e$ adjacent to $r$ 
	\State Set $\REFSTATE \gets \RefChosenSelectedSTATE$
\ElsIf{$\REFSTATE=\RefChosenSelectedSTATE$ (* There is already another chosen candidate $v$; all other candidates that approached $r$ so far were declined *)}
    \If{$\ID_u < \ID_v$}
        \State Add $\langle \ID_u, \ID_r, \DECLINED \rangle$ to $\SENDLIST(e)$ for every edge $e$ adjacent to $r$ 
    \ElsIf{$r$'s $\MLIST$ contains 
    $\langle \ID_v, \LOSES \rangle$    (* Some other referee generated a broadcast of the message $\langle \ID_v, \ID_u, \DECIDE \rangle$ and node $v$ subsequently broadcast that it lost its candidacy *)}
        \State Set $\Chosen \gets \ID_u$
        \State Add $\langle \ID_u, \ID_r, \APPROVED \rangle$ to $\SENDLIST(e)$ for every edge $e$ adjacent to $r$ 
    \ElsIf{$r$'s $\MLIST$ contains $\langle \ID_v, \ID_u, \DECIDE \rangle$ (* Some other referee generated a broadcast of the 
    dispute message *)}
        \State Set $\Contender \gets \ID_u$
        \State Set $\REFSTATE \gets \RefInDisputeSTATE$
    \Else
        \State Set $\Contender \gets \ID_u$
        \State Add $\langle\ID_v,\ID_u,\DECIDE\rangle$ to $\SENDLIST(e)$ for every edge $e$ adjacent to $r$ 
        \State Set $\REFSTATE \gets \RefInDisputeSTATE$
    \EndIf
\ElsIf{$\REFSTATE=\RefInDisputeSTATE$ (* A dispute is in progress between the current chosen $v$ and the current contender $w$, $\ID_w > \ID_v$ *)}
	\If{$\ID_u < \ID_w$}
        \State Add $\langle \ID_u, \ID_r, \DECLINED \rangle$ to $\SENDLIST(e)$ for every edge $e$ adjacent to $r$ 
	\Else
        \State Add $\langle \ID_w, \ID_r, \DECLINED \rangle$ to $\SENDLIST(e)$ for every edge $e$ adjacent to $r$
	    \State Set $\Contender \gets \ID_u$
        \State Add $\langle \ID_v, \ID_u, \DECIDE \rangle$ to $\SENDLIST(e)$ for every edge $e$ adjacent to $r$ 
	\EndIf
\EndIf
\end{algorithmic}
\end{algorithm}

\alglanguage{pseudocode}
\begin{algorithm}[H]
	\caption{Procedure $\RefereeDisputeResponse(r)$, run by a referee $r$ on receiving a reply to a 
	dispute about its chosen $v$ when it has a contender $u$.}
	\label{alg:referee-dispute-response}
	
\begin{algorithmic}[1]
\State Set $\Chosen \gets \ID_u$
\State Set $\Contender \gets -1$
\State Set $\REFSTATE \gets \RefChosenSelectedSTATE$
\State Add $\langle \ID_u, \ID_r, \APPROVED\rangle$ to $\SENDLIST(e)$ for every edge $e$ adjacent to $r$ 
\end{algorithmic}
\end{algorithm}

\end{document}